\documentclass[12pt,reqno]{amsart}
\usepackage{indentfirst, amssymb,amsmath,amsthm}
\usepackage[caption=false]{subfig}
\usepackage{graphicx}
\usepackage{tabulary}
\newtheorem{theo}{Theorem}[section]

\newtheorem{lem}{Lemma}

\newtheorem{cor}{Corollary}[section]

\newcommand{\be}{\begin{equation}}
\newcommand{\ee}{\end{equation}}
\newcommand{\beas}{\begin{eqnarray*}}
\newcommand{\eeas}{\end{eqnarray*}}
\newcommand{\bea}{\begin{eqnarray}}
\newcommand{\eea}{\end{eqnarray}}

\numberwithin{equation}{section}

\begin{document}

\setlength{\unitlength}{1mm} \baselineskip .45cm
\setcounter{page}{1}
\pagenumbering{arabic}
\title[ On Pseudo $B$-symmetric spacetimes]{ On Pseudo $B$-symmetric spacetimes and $f(\mathcal{R})$ gravity}

\author[  ]
{Young Jin Suh, Krishnendu De* and Uday Chand De  }

\address{  Department of Mathematics and RIRCM,
 Kyungpook National University,
Daegu-41566, South Korea. ORCID iD: https://orcid.org/0000-0003-0319-0738}
\email{yjsuh@knu.ac.kr}

\address
 { Department of Mathematics,
 Kabi Sukanta Mahavidyalaya,
The University of Burdwan.
Bhadreswar, P.O.-Angus, Hooghly,
Pin 712221, West Bengal, India. ORCID iD: https://orcid.org/0000-0001-6520-4520}
\email{krishnendu.de@outlook.in }
\address
{ Department of Pure Mathematics, University of Calcutta, West Bengal, India. ORCID iD: https://orcid.org/0000-0002-8990-4609}
\email {uc$_{-}$de@yahoo.com}

:
\thanks{$^{*}$ Corresponding author (krishnendu.de@outlook.in)\\
The first author was supported by the grant NRF-2018-R1D1A1B-05040381 from National Research Foundation of Korea.}
\maketitle
\begin {abstract}
This article delivers the characterization of a pseudo $B$ symmetric spacetimes and we illustrate that a pseudo $B$ symmetric spacetime admitting Codazzi type of $B$-tensor represents a perfect fluid spacetime and if this spacetime admits the time-like convergence criterion, then the pseudo $B$ symmetric spacetime fulfills cosmic strong energy criterion and contains pure matter. Besides, we find in a pseudo $B$ symmetric spacetime with Codazzi type of $B$-tensor the electric part of the Weyl tensor vanishes and has Riemann and Weyl compatible vector fields. Furthermore, it is established that the chosen spacetime with Codazzi type of $B$-tensor is conformally flat and represents a Robertson-Walker spacetime. Also, we calculate the scale factor $\varPsi (t)$ for these spacetimes in a spatially flat Robertson-Walker spacetime. Finally, we study the impact of this spacetime under $f(R)$ gravity scenario and deduce several energy conditions by considering a new model $f\left(\mathcal{R}\right)= e^{(\alpha \mathcal{R})}-ln(\beta \mathcal{R})$ in which $\alpha$ and $\beta$ are positive constants.
\end {abstract}

\section{\textbf{Introduction}}
\hspace{0.3cm}

It is widely accepted that symmetric spaces are crucial to differential geometry. Each symmetric space has a unique geometry, including hyperbolic, elliptical, and Euclidean geometry. We can naturally extend classical harmonic analysis on spheres using symmetric spaces. These spaces have application in diverse areas of mathematics, including bounded symmetric fields, Grassmannian and compact Lie groups. Cartan \cite{cartan} initiated the study of symmetric Riemannian spaces in the second decade of twentieth century, and he acquired a classification of this space.\par

However, symmetric spaces can be viewed from various angles and possess many similar characteristics. Riemannian manifolds with point reflection, parallel curvature tensors, homogeneous spaces with unique isotropy, holonomy, special Killing vector fields, or specific Lie group involution are some possible interpretations for them. Broadly speaking, symmetric spaces are differential manifolds endowed with an evolution at each point, where each point is a fixed isolated point meeting certain requirements. Symmetric Riemannian spaces are just called symmetric spaces when there is no chance of confusion.\par

Suppose $\nabla$ is the Levi-Civita connection of a Riemannian manifold $(M,g)$. According to \cite{cartan}, for $\nabla \mathcal{R}=0$, $M$ is considered as a locally symmetric manifold, in which $\mathcal{R}$ represents the Riemannian curvature tensor of $M$.\par

A manifold $M$ is named Ricci symmetric if $\nabla_{k}\mathcal{R}_{ij}=0$, in which $\mathcal{R}_{ij}$ is the Ricci tensor of $M$. In \cite{mcc}, Chaki introduced pseudo Ricci symmetric manifolds as a generalization of these manifolds. If the Ricci tensor of a non-flat Riemannian manifold $(M^{n},g)$ of dimension $n(>2)$ satisfies the condition
\begin{equation}\label{1.1}
  \nabla _{k}\mathcal{R}_{ij}=2A_{k}\mathcal{R}_{ij}+A_{i}\mathcal{R}_{kj}
+A_{j}\mathcal{R}_{ik},
\end{equation}
then $M$ is considered as pseudo Ricci symmetric manifold, where $\nabla_{k}$ indicates the covariant differentiation with respect to $g$ and $A_{i}$ is a non-zero 1-form. The associated 1-form of $M$ is denoted by $A_{i}$. This manifold $M$ yields a symmetric manifold of Cartan sense if $A_{i}=0,$. A pseudo Ricci symmetric manifold of dimension $n$ is indicated by $(PRS)_{n}$.\par

Weakly Ricci symmetric manifolds were introduced in 1993 by Tamassy and Binh \cite{tamsbin93}. One may note that a specific instance of a weakly Ricci symmetric manifold is a pseudo Ricci symmetric manifold. Suh and Mantica introduced pseudo $Z$ symmetric manifolds denoted by $(PZS)_{n}$, in a recent paper \cite{mantica1}. It was an extension of the concept of pseudo projective-Ricci symmetric manifolds \cite{mcc1}, or pseudo Ricci symmetric manifolds \cite{mcc}.\par

A special second-order, symmetric tensor field $Z$ is called generalized $Z$ tensor and defined by
\begin{equation}\label{1.2}
  Z_{ij}=\mathcal{R}_{ij}+\phi g_{ij},
\end{equation}
in which the arbitrary scalar function is denoted by $\phi $. Transvecting (\ref{1.2}) with $g^{ij}$ yields the scalar $Z$ as:
\begin{equation}\label{1.3}
  Z=n\phi +\mathcal{R}.
\end{equation}

Lorentzian manifolds that admit a globally time-like vector field, physically known as spacetime. A generalized Robertson-Walker (GRW) spacetime (\cite{alias1}, \cite{bychen}) is a Lorentzian manifold of dimension $n$ $\left(n\geq4\right)$ which can be expressed by a warped product $-I\times_{\varPsi^{2}}\stackrel{\ast}{M}$, where $\stackrel{\ast}{M}$ is the Riemannian manifold of dimension $(n-1)$, $I \in \mathbb{R}$ (set of real numbers) is an open interval, and $\varPsi>0$ denotes the scale factor. In particular, a GRW spacetime becomes a Robertson-Walker (RW) spacetime if we assume that $\stackrel{\ast}{M}$ is a 3-dimensional Riemannian manifold and is of constant sectional curvature.\par

A $4$-dimensional Lorentzian manifold $M^{4}$ is called a perfect fluid spacetime (PFS) if the non-zero Ricci tensor $\mathcal{R}_{lk}$ obeys
\begin{equation}\label{1.1a}
\mathcal{R}_{lk}=\alpha g_{lk}+\beta u_{l}u_{k},
\end{equation}
where $\alpha$, $\beta$ are scalars (not simultaneously zero) and the velocity vector $u_{k}$ is unit time-like, that is, $u_{k}u^{k}=-1$, $u^{k}=g^{lk}u_{l}$. The matter field in general relativity (GR) is denoted by $T_{lk}$, which is called the energy momentum tensor. The energy momentum tensor \cite{Neil} for a PFS has the following form
\begin{equation}\label{1.2a}
T_{lk}=\left(\sigma+p\right)u_{l}u_{k}+p g_{lk},
\end{equation}
$p$, $\sigma$ being the isotropic pressure and the energy density.\par

If $\sigma=p$, then the PFS is a stiff matter fluid. Also, if $\sigma+p=0$, $\sigma=3p$ and $p=0$, then the PFS is called the dark energy epoch of the Universe, radiation era and dust matter fluid \cite{chav}, respectively. Without a cosmological constant, the Einstein's field equations are given by
\begin{equation}\label{1.3a}
\kappa T_{lk}-\mathcal{R}_{lk}+\dfrac{\mathcal{R}}{2}\,g_{lk}=0
\end{equation}
in which $\kappa$ denotes gravitational constant and $\mathcal{R}$ is the Ricci scalar.\par

Equations \eqref{1.1a}, \eqref{1.2a} and \eqref{1.3a} together reveal \cite{manticamolinaride}
\begin{equation}\label{1.4a}
	\alpha=\kappa\left(\dfrac{\sigma-p}{2}\right)\quad\mathrm{and}\quad\beta=\kappa\left(\sigma+p\right).
\end{equation}
In $M^{4}$, the Weyl conformal curvature tensor $\mathcal{C}_{lijk}$ is written by
\begin{equation}\label{1.5aa}
	\mathcal{C}_{lijk}=\mathcal{R}_{lijk}-\dfrac{1}{2}\left\{g_{ij}\mathcal{R}_{lk}-g_{ik}\mathcal{R}_{lj}
+g_{lk}\mathcal{R}_{ij}-g_{lj}\mathcal{R}_{ik}\right\}
+\dfrac{\mathcal{R}}{6}\left\{g_{lk}g_{ij}-g_{lj}g_{ik}\right\}
\end{equation}
in which $\mathcal{R}_{lijk}$ is the $(0,4)$ type curvature tensor.\par

It is well-known that \cite{E49}
\begin{equation}\label{1.6aa}
	\nabla_{l}\mathcal{C}^{l}_{ijk}=\dfrac{1}{2}\left\{\left(\nabla_{k}\mathcal{R}_{ij}-\nabla_{j}\mathcal{R}_{ik}\right)
-\dfrac{1}{6}\left(g_{ij}\nabla_{k}\mathcal{R}-g_{ik}\nabla_{j}\mathcal{R}\right)\right\}.
\end{equation}

In a spacetime of dimension 4, Weyl tensor is represented by two symmetric tensors having ten independent components. If a vector field $A$ satisfies the condition $A^{j}A_{j}=-1$, then the magnetic and electric components of the Weyl tensor are given by

\begin{align}\label{17}
H_{ij}=A^kA^l\tilde{C}_{kijl},~~~~~E_{ij}=A^kA^lC_{kijl}
\end{align}	
where $\tilde{C}_{kijl}=\frac{1}{2}{\varepsilon}_{kimn}C^{mn}_{jl}$	is the dual \cite{Bert}. In this case, $E_{ij}$ and $H_{ij}$ are traceless and symmetric. The conditions $E_{ij}A^i=0$ and $H_{ij}A^i=0$ are satisfied by these tensors. As a result, they fully characterize the Weyl tensor and each have five independent components.\par

A (0,2) symmetric tensor $B_{ij}$ is introduced in \cite{smup}, in the following way:
\begin{equation}\label{1.4}
  B_{lk}=a\mathcal{R}_{lk}+b\mathcal{R} g_{lk},
\end{equation}
in which $a$ and $b$ (non-zero) indicate arbitrary scalar functions. Transvecting (\ref{1.4}) with $g^{lk}$ produces the scalar $B$ in the following way:
\begin{equation}\label{1.5}
  B=(a+nb)\mathcal{R}.
\end{equation}
For a B-flat manifold, that is, $B_{lk}=0$, the manifold reduces to $\mathcal{R}_{lk}=-\frac{b\mathcal{R}}{a}g_{lk}$, that is, an Einstein manifold \cite{besse}.\par

Now we choose a $B$-flat spacetime and a conformally flat spacetime with vanishing $B$-tensor and establish the following:
\begin{theo}\label{the1}
  A $B$-flat GRW spacetime is a PFS and a $B$-flat PFS is a GRW spacetime.
\end{theo}

\begin{theo}\label{the2}
  Let a conformally flat spacetime admit vanishing $B$-tensor. Then the spacetime is either locally isometric to Minkowski spacetime, or a de-sitter or anti de-sitter spacetime.
\end{theo}

For a B-recurrent manifold, that is, $\nabla_{k}B_{ij}=\lambda_{k}B_{ij},$ the manifold turns into a generalized
Ricci-recurrent manifold \cite{deguhaka} which is identical to
\begin{equation*}
  \nabla_{k}\mathcal{R}_{ij}=\mu_{k}\mathcal{R}_{ij}+(n-1)\gamma_{k}g_{ij},
\end{equation*}
in which $\gamma_{k}=-(\mathcal{R} \nabla_{k}b
+\nabla_{k}\mathcal{R} b)+\lambda_{k}b \mathcal{R}$ and $\mu_{k}=-\frac{\nabla_{k}a}{a}+a\lambda_{k}$. The manifold provides a Ricci recurrent manifold for $\mu_{k}=1$ and $\gamma_{k}=0$.\par

With the cosmological constant $\lambda$, Einstein equation \cite{de felis} and energy momentum tensor $T_{lk}$ may be given by
\begin{equation}\label{1.5a}
  \frac{1}{a}B_{lk}=\kappa T_{lk},
\end{equation}
where $\frac{b\mathcal{R}}{a}=-\frac{1}{2}\mathcal{R}+\lambda, a\neq0$. Thus, the tensor $\frac{1}{a}$ times of $B_{lk}$ may be considered as a generalized Einstein gravitational tensor with arbitrary scalar function $\frac{b\mathcal{R}}{a}.$\par

Constraints on the $B$-tensor are determined by various situations on the energy momentum tensor. An Einstein space with $\lambda=\frac{n-2}{2n}\mathcal{R}$ is determined by the vacuum solution $B=0$. A spacetime with conserved energy-momentum density is described by the relation $\nabla_{i}B_{lk}=0$, which is deduced by the conservation of total energy momentum tensor $(\nabla_{k}T_{lk}=0)$.\par

In 1952, Patterson \cite{P52} invented the notion of Ricci-recurrent manifolds. In 1995, De et al. \cite{deguhaka} introduced the concept of generalized Ricci-recurrent manifolds. The significance of the Generalized Ricci-recurrent structure and its interaction with the modified $f(\mathcal{R})$-theory \cite{ade1} and the modified Gauss-Bonnet $f(\mathcal{R},G)$-theory \cite{ade2} are well established. On the other hand, a condition for which a pseudo symmetric spacetime would be a PFS was recently established by Zhao et al.\cite{kdez}). Pseudo $Z$-symmetric spacetimes have been investigated by Mantica and Suh \cite{ManSuh} and Ozen\cite{ozen} has studied $m$-Projectively flat spacetimes. Also, in \cite{fku} $\Psi$-conformally symmetric spacetimes have been investigated. Moreover, in \cite{kdeu1}, we have studied $\psi$-conharmonically symmetric spacetime. As well, many authors have looked at the spacetime of general relativity in various methods; for additional information, see ( \cite{dh}, \cite{dh1}, \cite{gul2}).\par

Motivated by the importance of the $B$ tensor in relativity and previous works we investigate a novel spacetime named pseudo $B$ symmetric spacetimes. A 4-dimensional Lorentzian manifold is called pseudo $B$ symmetric spacetime and denoted by $(PBS)_{4},$ if the non-zero tensor $B$ obeys the relation
\begin{equation}\label{1.6}
  \nabla _{k}B_{ij}=2A_{k}B_{ij}+A_{i}B_{kj} +A_{j}B_{ik}.
\end{equation}
where the associated vector $A_{i}$ is unit time-like.\par

It is evident that the $(PBS)_{4}$ turns into $(PZS)_{4} $(\cite{mantica1},\cite{ManSuh}) for $a=1$ and $b=\frac{\phi}{\mathcal{R}}$ and the $(PBS)_{4}$ reduces to pseudo Ricci-symmetric manifolds\cite{mcc} for $a=1$ and $b=0$.\par

In \cite{gray}, Gray introduced the concept of Codazzi type of Ricci tensor. A manifold satisfies Codazzi type of Ricci tensor if its non-zero Ricci tensor $\mathcal{R}_{ij}$ obeys
\begin{equation}\label{1.8}
  \nabla _{k}\mathcal{R}_{ij}=\nabla _{j}\mathcal{R}_{ik}.
\end{equation}
The study of different spacetimes with Codazzi type of Ricci tensor have been found in (\cite{kde3}, \cite{mm5}) and references there in.\par

Now let us denote a $(PBS)_{4}$ spacetime with Codazzi type of B-tensor by $\mathcal{N}$ and we prove the following theorems:
\begin{theo}\label{th3.1}
The spacetime $\mathcal{N}$ is a PFS.
\end{theo}
\begin{theo}\label{th3.3}
Let $\mathcal{N}$ admit the time-like convergence condition. Then, $\mathcal{N}$ satisfies cosmic strong energy condition and contains pure matter.
\end{theo}

\begin{theo}\label{th3.4}
    $\mathcal{N}$ is a GRW spacetime.
\end{theo}

\begin{theo}\label{th3.5}
  In $\mathcal{N}$, the electric part of the Weyl tensor vanishes and has Riemann compatible vector field $A_{i}$.
\end{theo}

\begin{theo}\label{th3.7}
 $\mathcal{N}$ is conformally flat and belongs to Petrov type O and a RW spacetime.
\end{theo}

Einstein's field equations are not adequate to determine the late-time inflation of the cosmos without assuming the existence of certain unseen components that could account for the origins of dark matter and dark energy. This is the primary source of motivation for the extension to acquire higher order gravity's field equations.\par

In GR, to investigate black holes and wormholes, energy conditions (ECs) are prime tools in many modified gravities (\cite{BIBY17}, \cite{HE73}). The energy conditions were provided in \cite{RBB92} by applying Raychaudhuri equations and demonstrate the character of gravity with the positivity condition $R_{lk}u^{l}u^{k}\geq0$, where $u^{l}$ is the null vector. The very last condition is equal to the null energy condition (NEC) $T_{lk}u^{l}u^{k}\geq0$. Moreover, for any time-like vector $u^{l}$, the weak energy condition (WEC) asserts that $T_{lk}u^{l}u^{k}\geq0$. In addition, a space-time obeys the dominant energy condition (DEC) if $T_{lk}u^{l}v^{k}\geq0$ holds for every two co-oriented time-like vectors $u$ and $v$ and strong energy condition (SEC) \cite{dug} if $R_{lk}u^{l}u^{k}\geq0$ holds for all time-like vectors $u$.\par

Interestingly, the idea of $f(\mathcal{R})$-gravity appears as a spontaneous extension of Einstein's gravity theory. The function $f(\mathcal{R})$, $\mathcal{R}$ denotes the Ricci scalar modifies the Hilbert-Einstein action term. The above theory was invented by Buchdahl \cite{hab}, and Starobinsky \cite{aas} has demonstrated its validity through research on cosmic inflation. Through the introduction of certain couplings between the geometrical quantities and the matter sector, the $f(\mathcal{R})$ theory of gravity has been further generalized. The non-minimal coupling between the matter lagrangian density and the curvature invariant has been proved in \cite{har}, which is known as the $f(\mathcal{R},L_m)$ theory of gravity. The corresponding Lagrangian can be modified by incorporating an analytic function of $T_{lk}T^{lk}$ in this generalization procedure for the $f(\mathcal{R},L_m)$ theory. $f(\mathcal{R},T^2)$ gravity or energy-momentum squared gravity is the result of selecting the corresponding Lagrangian. In 2014, Katirci and Kavuk \cite{kat} originally put forward this theory, which allows the existence of a term in the action functional that is proportional to $T_{lk}T^{lk}$. The $f(\mathcal{R},G)$-gravity theory \cite{EMOS10} was one of these modified theories. It was developed by changing the previous Ricci scalar $\mathcal{R}$ by a function of $\mathcal{R}$ and $G$, the Gauss-Bonnet invariant. The $f(\mathcal{R},T)$-gravity theory, discovered by Harko et al. \cite{HLNO11}, was another modified theory. This is an extension of $f\left(\mathcal{R}\right)$-gravity in which the trace $T$ of the energy momentum tensor is directly linked to any arbitrary function of $\mathcal{R}$. Several functional forms of $f(\mathcal{R})$ were previously provided in the following works: (\cite{cap}, \cite{cap2}, \cite{cap3}, \cite{ade}, \cite{ade1}, \cite{kde},\cite{kde1}) which shows that this modified theory has a number of cosmological applications.\par

The literature mentioned above makes it very clear that more attention needs to be paid to $f(\mathcal{R})$ gravity, and there are still a lot of unanswered questions. Inspired by the foregoing investigations, this article is focused to examine a $(PBS)_{4}$ spacetime with Codazzi type of $B$-tensor satisfying $f(\mathcal{R})$ gravity and explain different Energy conditions (ECs). In the literature, there are few exponential models examined by S. I. Kruglov and
S. D. Odintsov (\cite{kr}, \cite{od}). The functional form $f\left(\mathcal{R}\right)=\mathcal{R}+\alpha ln(\beta \mathcal{R})$ is investigated in \cite{gi}, in which $\alpha$ and $\beta$ are constants. Also, the Logarithm
function is continuous and differentiable when $\beta \mathcal{R}> 0$. For this reason, we choose the new model $f\left(\mathcal{R}\right)= e^{(\alpha \mathcal{R})}-ln(\beta \mathcal{R})$ in which $\alpha$ and $\beta$ are positive constants, to explain different energy conditions.\par

\section{\textsf{Proof of the Main Results}}

\vspace{.6cm}
{\bf{Proof of the Theorem \ref{the1}}}

Let us choose a $B$-flat spacetime. Then, equation \eqref{1.4} entails
\begin{equation}\label{bb1}
  \mathcal{R}_{ij}=-\frac{b}{a}\mathcal{R} g_{ij}.
\end{equation}
Contracting the previous equation yields
\begin{equation}\label{bb2}
  \mathcal{R}=0,\;\;\text{or}\;\;a=-4b.
\end{equation}
From equation \eqref{bb1} we say that it is an Einstein spacetime and hence from \eqref{1.6aa} we infer $\nabla_{h}C^{h}_{ijk}=0$. In \cite{Mantica5}, it is established that a GRW spacetime obeys $\nabla_{h}C^{h}_{ijk}=0$ if and only if the spacetime is a PFS. $\square$\par

\vspace{.4cm}
{\bf{Proof of the Theorem \ref{the2}}}

For a conformally flat spacetime using equation \eqref{1.5aa}, the curvature tensor $R_{hijk}$ is written as
 \begin{eqnarray}\label{bb3}
  R_{hijk}&=& \frac{1}{2}(g_{hk}R_{ij}-g_{hj}R_{ik}+g_{ij}R_{hk}-g_{ik}R_{hj})\nonumber\\&&
  -\frac{\mathcal{R}}{6}\{g_{hk}g_{ij}-g_{hj}g_{ik}\},
\end{eqnarray}
Now, we choose a conformally flat spacetime with a vanishing $B$-tensor. Then using equation \eqref{bb1} in equation \eqref{bb3}, we get
\begin{eqnarray}\label{bb4bbbb}
  R_{hijk}&=& \frac{1}{2}\{-\frac{b}{a}\mathcal{R} g_{hk}g_{ij}+\frac{b}{a}\mathcal{R}g_{hj}g_{ik}\nonumber\\&&
  -\frac{b}{a}\mathcal{R} g_{ij}g_{hk}+\frac{b}{a}\mathcal{R}g_{ik}g_{hj}\}\nonumber\\&&
  -\frac{\mathcal{R}}{6}\{g_{hk}g_{ij}-g_{hj}g_{ik}\}\nonumber
  \end{eqnarray}
which implies
\begin{equation}\label{bb4}
  R_{hijk}=-\mathcal{R}\big(\frac{6b+a}{6a}\big)\{g_{hk}g_{ij}-g_{hj}g_{ik}\}.
\end{equation}
Since here $B$ tenor vanishes, then by previous theorem we have either $\mathcal{R}=0$, or $a=-4b$.\par

If $\mathcal{R}=0$, then equation \eqref{bb4} tells that the spacetime has vanishing sectional curvature. Therefore, the spacetime and Minkowski spacetime are locally isometric (\cite{dug}, p. 67).\par

If $a=-4b$, then the equation \eqref{bb4} provides
\begin{equation}\label{bb5}
  R_{hijk}=\frac{\mathcal{R}}{12}\{g_{hk}g_{ij}-g_{hj}g_{ik}\}.
\end{equation}
Hence, the spacetime represents a spacetime of constant curvature. In Lorentzian settings, the spaces of constant curvature are well classified in (\cite{HE73}, pages 124-131) and the spacetime reduces to the de Sitter spacetime, whenever $\mathcal{R}>0$; anti de Sitter spacetime, whenever $\mathcal{R}<0$. $\square$\par
\vspace{.4cm}
{\bf{Proof of the Theorem \ref{th3.1}}}

Considering the spacetime $\mathcal{N}$, we find from \eqref{1.6}
\begin{equation}\label{4.4}
  \nabla_{k}B_{jl}-\nabla_{j}B_{kl}=A_{k}B_{jl}-A_{j}B_{kl}.
\end{equation}

By the hypothesis, $B$ is of Codazzi type. Therefore, equation \eqref{4.4} infers
\begin{equation}\label{4.5}
  A_{k}B_{jl}-A_{j}B_{kl}=0.
\end{equation}
Multiplying the foregoing equation by $A^{k}$ gives
\begin{equation}\label{4.6}
  B_{jl}=-A_{j}A^{k}B_{kl},
\end{equation}
since $A_{k}$ is a unit time-like vector, that is, $A_{k}A^{k}=-1$.\par

Again, multiplying equation \eqref{4.5}) by $g^{jl}$, we reveal
\begin{equation}\label{4.7}
  A_{k}B-A^{l}B_{kl}=0.
\end{equation}
Using equation \eqref{4.7} in equation \eqref{4.6} yields
\begin{equation}\label{4.8}
  B_{jl}=-A_{j}A_{l} B.
\end{equation}
Hence, from equation \eqref{1.4}, we have
\begin{equation}\label{4.9}
  \mathcal{R}_{jl}=-\frac{b \mathcal{R}}{a}g_{jl}-\frac{(a+4b)}{a}\mathcal{R} A_{j} A_{l},
\end{equation}

which ends the proof.\par
\vspace{.4cm}

In light of equations \eqref{1.1a}, \eqref{1.4a} and \eqref{4.9}, we obtain
\begin{equation}\label{4.11}
\kappa\left(\dfrac{\sigma-p}{2}\right)=-\frac{b \mathcal{R}}{a}
\end{equation}
and
\begin{equation}\label{4.12}
\kappa\left(\sigma+p\right)=-\frac{(a+4b)}{a}\mathcal{R}.
\end{equation}
Equations \eqref{4.11} and \eqref{4.12} together give
\begin{equation}\label{4.13}
\dfrac{p}{\sigma}=\dfrac{a+2b}{a+6b}\,.
\end{equation}
We observe that \eqref{4.13} implies $p=0$ for $a=-2b$, $\sigma=3p$ for $a=0$ which is not possible and $\sigma+p=0$ for $a=-4b$, respectively. Thus, we provide:
\begin{cor}
The spacetime $\mathcal{N}$ represents a
\begin{enumerate}
\item the equation of state of the form \eqref{4.13},
\item dust matter fluid for $a=-2b$,
\item dark energy epoch of the Universe for $a=-4b$,
\item but, can not admit radiation era.
\end{enumerate}
\end{cor}
\vspace{.4cm}
{\bf{Proof of the Theorem \ref{th3.3}}}

Solving the equations \eqref{4.11} and \eqref{4.12}, we get
\begin{equation}\label{4.14}
  p=-\frac{1}{2 a\kappa }(a+2b)\mathcal{R}
\end{equation}
and
\begin{equation}\label{4.15}
  \sigma=-\frac{1}{2 a\kappa }(a+6b)\mathcal{R}.
\end{equation}

A time-like convergence criterion for each time-like vector field $A_{j}$ of a spacetime is defined as $\mathcal{R}_{jk} A^{j}A^{k} > 0$ \cite{sac}. Let $\mathcal{N}$ admit the time-like convergence criterion. Then transvecting equation \eqref{4.9} by $ A^{j}A^{l}$, we acquire
\begin{equation}\label{4.16}
  A^{j}A^{l} \mathcal{R}_{jl}=-\frac{a+3b}{a}\mathcal{R}.
\end{equation}
By hypothesis, $A_{j}$ is a time-like vector field and the convergence condition $\mathcal{R}_{jk} A^{j}A^{k}>0$ holds. Hence, we get $\mathcal{R}<0$.\par

Therefore, Using equations \eqref{4.14} and \eqref{4.15}, we provide that $p>0$ and $\sigma>0$. Hence, we obtain $3p+\sigma >0$ which entails that $\mathcal{N}$ satisfies cosmic strong energy criterion.\par

In $\mathcal{N}$, we acquire that $\mathcal{R}<0$. Further, equation \eqref{4.16} infers that $\sigma>0$. Hence, we write, $\mathcal{N}$ contains pure matter.$\square$\par
\vspace{.4cm}
{\bf{Proof of the Theorem \ref{th3.4}}}

Again, the equation \eqref{1.4} yields
\begin{equation}\label{4.18}
  \nabla_{k}\mathcal{R}_{ij}-\nabla_{j}\mathcal{R}_{ik}=\nabla_{k}B_{ij}-\nabla_{j}B_{ik}.
\end{equation}

By hypothesis, $\nabla_{k}B_{ij}=\nabla_{j}B_{ik}$ and hence, from equation \eqref{4.18}, we acquire
\begin{equation}\label{4.19}
  \nabla_{k}\mathcal{R}_{ij}-\nabla_{j}\mathcal{R}_{ik}=0.
\end{equation}
Therefore, the last equation infers $\mathcal{R} =$ constant and hence from equation \eqref{1.6aa}, we provide $\nabla_{h}C^{h}_{ijk}=0$.\par

It is well known that a PFS with $\nabla_{h}C^{h}_{ijk}=0$ is a GRW spacetime \cite{Mantica5} and hence $\mathcal{N}$ reduces to a GRW spacetime.$\square$\par
\vspace{.4cm}
{\bf{Proof of the Theorem \ref{th3.5}}}

Now, from the last Theorem we see that $\mathcal{N}$ is a $GRW$ spacetime.\par

In \cite{survey}, Mantica and Molinari have shown the following:\par
{\bf{Theorem I.}}(\cite{survey})
 An $n$ ($n \geq 3$) dimensional Lorentzian manifold is a $GRW$ spacetime if and only if it admits a unit time-like torse-forming vector field: $\nabla_{i}v_{j}=\Psi (g_{ij}+v_{j}v_{i})$, which is also an eigenvector of the Ricci tensor.\\

Therefore, from Theorem I we have
\begin{equation}\label{n1}
  \nabla_{j}A_{i}=\Psi (g_{ij}+A_{i}A_{j}).
\end{equation}
The covariant derivative of equation \eqref{n1} reveals

\begin{align}\label{n2}
\nabla_{k}\nabla_{j}A_{i}={\Psi}_k(g_{ij}+A_iA_j)
+{\Psi}\big[{\Psi}(g_{ik}+A_iA_k)A_j+{\Psi}(g_{jk}+A_jA_k)A_i\big]
\end{align}
where ${\Psi}_k=\nabla_{k} \Psi$.\par

Assuming $w_k={\Psi}_k-{\Psi}^2A_k$, the last equation easily yields

\begin{align}\label{n3}
\nabla_{k}\nabla_{j}A_{i}-\nabla_{j}\nabla_{k}A_{i}=w_k(g_{ij}+A_iA_j)-w_j(g_{ik}+A_iA_k).
\end{align}
If we consider $w_k=fA_k$, then equation \eqref{n3} provides

 \begin{align}\label{n6}
 \mathcal{R}^h_{ijk}A_h=f(g_{ij}A_k-g_{ik}A_j).
  \end{align}	
	
In dimension 4, we know that
\begin{align}\label{n7}
E_{ij}= \mathcal{R}_{kijl}A^kA^l-\frac{1}{2}\{g_{ij}\mathcal{R}_{kl}-g_{il}\mathcal{R}_{kj}
+g_{kl}\mathcal{R}_{ij}-g_{kj}\mathcal{R}_{il}\}A^kA^l\nonumber\\
+\frac{\mathcal{R}}{6}(g_{ij}g_{kl}-g_{il}g_{jk})A^kA^l.
\end{align}

Using equations \eqref{4.9} and \eqref{n7}, we acquire
\begin{align}\label{n9}
E_{ij}= \mathcal{R}_{kijl}A^kA^l-\frac{a+3b}{3a}\mathcal{\mathcal{R}}\{g_{ij}+A_iA_j\}.
\end{align}

Using equation \eqref{n6} in \eqref{n9}, we provide

\begin{align}\label{n10}
E_{ij}=-(f+\frac{a+3b}{3a}\mathcal{\mathcal{R}})\{g_{ij}+A_iA_j\}.
\end{align}
Again, from equation \eqref{n6}, we infer

\begin{align}\label{n11}
A^h\mathcal{R}_{hk}=3fA_k.
\end{align}

Thus, using the equation \eqref{4.9}, the above equation gives

\begin{align}\label{n12}
f=-\frac{a+3b}{3a}\mathcal{R}.
\end{align}

Therefore, from equations \eqref{n10} and \eqref{n12}, we conclude that $E=0$.\par

Using equation \eqref{n6}, we acquire
 \begin{align}\label{p1}
 \mathcal{R}_{kjim}A^{m}=f\{g_{ij}A_k-g_{ik}A_j\}.
  \end{align}
Making use of the previous equation, we provide

\begin{align}\label{p2}
\{(\mathcal{R}_{jklm}A^{m})A_i+(\mathcal{R}_{kilm}A^{m})A_j+(\mathcal{R}_{ijlm}A^{m})A_k\}=0.
\end{align}
Considering the equation \eqref{p2}, we reveal
\begin{align}\label{p3}
\{\mathcal{R}_{jklm}A_{i}+\mathcal{R}_{kilm}A_j+\mathcal{R}_{ijlm}A_k\}A^m=0.
\end{align}
This shows that $\mathcal{N}$ has Riemann compatible vector field $A_{i}$.$\square$\par
\vspace{.4cm}
{\bf{Proof of the Theorem \ref{th3.7}}}

It is well circulated that $A_{j}$ is Riemann compatible \cite{Mantica3} if and only if it is Weyl compatible and
	\begin{align}\label{p4}
	A_{[j}R^m_{k]}A_m=0.
	\end{align}
Hence, from the last Theorem we conclude that vector field $A_{i}$ is Weyl compatible and therefore, $H=0$ \cite{Mantica3}.
Again, in a Lorentzian manifold of dimension $n$ $(n>3)$, the Weyl tensor $C$ is purely electric if and only if for every coordinate domain, we reveal \cite{her}
\begin{align}\label{p5}
\{C^{m}_{jkl}A_{i}+C^{m}_{kil}A_j+C^{m}_{ijl}A_k\}=0.
\end{align}
 Obviously, the equation \eqref{p5} is satisfied by the Weyl compatible vector field $A_{i}$. Therefore, we acquire $H=0$ and $E=0$. At last, we see that $\mathcal{N}$ is conformally flat and therefore, $\mathcal{N}$ is of Petrov type O.\par

Considering the Theorem \ref{th3.4}, we say that $\mathcal{N}$ is a GRW spacetime and also conformally flat. We know that a conformally flat GRW spacetime is a RW spacetime \cite{bro}.$\square$\par
\vspace{.4cm}
From the last Theorem we find that $\mathcal{N}$ is a RW spacetime.\par
A spatially flat RW spacetime provides
\begin{equation}\label{s1}
  ds^{2}=-dt^{2}+\varPsi ^{2}(t)[dr^2+r^2(d\theta^2+sin^2\theta d\phi ^{2})].
\end{equation}
The non-zero components of the Ricci tensor are
\begin{equation}\label{s11}
  R_{00}=-3\frac{\ddot{\varPsi}}{\varPsi};\,\,
  R_{11}=\;R_{22}=\;R_{33}=\varPsi^2[2(\frac{\dot{\varPsi}}{\varPsi})^2+\frac{\ddot{\varPsi}}{\varPsi}],
\end{equation}
and the scalar curvature is then
\begin{equation}\label{s2}
  R=6[\frac{\ddot{\varPsi}}{\varPsi}+(\frac{\dot{\varPsi}}{\varPsi})^2].
\end{equation}

Equations \eqref{1.3a} and \eqref{1.5a} jointly produce
\begin{equation}\label{s3}
\frac{1}{a}B_{lk}=\mathcal{R}_{lk}-\dfrac{\mathcal{R}}{2}\,g_{lk}.
\end{equation}
In a spacetime $\mathcal{N}$, from equation \eqref{4.7}, we acquire
\begin{equation}\label{s4}
  A^{k}B_{lk}=A_{l}B
\end{equation}
such that
\begin{equation}\label{s5}
  A_{l}=(-1,0,0,0),\;\;A^{l}A_{l}=-1.
\end{equation}
From the last two equations, we obtain
\begin{equation}\label{s6}
  A^{0}B_{00}=BA_{0}.
\end{equation}
Now using equations \eqref{1.5}, \eqref{s11}, \eqref{s3} in \eqref{s6}, we get
\begin{equation}\label{s7}
  \mathcal{R}_{00}=\frac{a+8b}{a}\mathcal{R}.
\end{equation}
Using equations \eqref{s11} and \eqref{s2}, the last equation gives
\begin{equation}\label{s8}
  \frac{(2a+16b)}{(3a+16b)}\frac{\dot{\varPsi}}{\varPsi}=\frac{\ddot{\varPsi}}{\dot{\varPsi}}.
\end{equation}
Solving the above equation, we reveal
\begin{equation}\label{s9}
  \varPsi (t)=\frac{1}{(\frac{a_{1}}{(a_{1}-1)
  (c_{1}t+c_{2})})^{\frac{a_{1}}{a_{1}-1}}}.
\end{equation}
where $a_{1}=\frac{(2a+16b)}{(3a+16b)}$.\par
 Hence, we state:
 \begin{cor}\label{th3.9}
  In a spatially flat spacetime $\mathcal{N}$, the scale factor $\varPsi (t)$ is described by the equation \eqref{s9}.
\end{cor}

\section{Application of $\mathcal{N}$ spacetime in $f(\mathcal{R})$ gravity}

Now, we will look into a few characteristics of $\mathcal{N}$ in $f(\mathcal{R})$ gravity.\par
Choose the Einstein-Hilbert action
\begin{align}
S=\frac{1}{2k^2}\int{d^4x{\sqrt{-g}}f(\mathcal{R})}+\int{d^4x{\sqrt{-g}}L_{m}},
\nonumber
\end{align}
where  $L_m$ is the matter Lagrangian density described as
\begin{align}
T_{ij}=-\frac{2}{\sqrt{-g}}\frac{\delta (\sqrt{-g}L_m)}{\delta g^{ij}}.
\nonumber
\end{align}
In this case, $\kappa^{2}=8{\pi}G$, $G$ indicates the Newton's constant and hence the modified formula written as \cite{Sotiriou}
\begin{align}
S=\frac{1}{2k^2}\int{d^4x{\sqrt{-g}}f(\mathcal{R})}.
\nonumber
\end{align}

The field equations can be expressed in the following shape by applying the variation with $g^{ij}$
\begin{align}\label{36}
(g_{ij}{\square}-\nabla_{i}\nabla_{j})f_{\mathcal{R}}(\mathcal{R})-\frac{f(\mathcal{R})}{2}g_{ij}
+f_{\mathcal{R}}(\mathcal{R})\mathcal{R}_{ij}
=k^2T_{ij}
\end{align}
where $\square$ indicates the D'Alembert's operator and $f_{\mathcal{R}}(\mathcal{R})$, is the derivative with respect to $\mathcal{R}$.\par

Taking trace of equation \eqref{36} provides
\begin{align}\label{37}
-2f(\mathcal{R})+\mathcal{R}f_{\mathcal{R}}(\mathcal{R})+3{\square}f_{\mathcal{R}}(\mathcal{R})=k^2T.
\end{align}

Subtracting the expression $\frac{\mathcal{R}f_{\mathcal{R}}(\mathcal{R})}{2} g_{ij}$ from \eqref{36}, we infer
\begin{align}\label{38}
f_{\mathcal{R}}(\mathcal{R})\mathcal{R}_{ij}-\frac{\mathcal{R}f_{\mathcal{R}}(\mathcal{R})}{2}g_{ij}
=k^2T_{ij}+k^2T^{(curve)}_{ij}
\end{align}
such that
\begin{align}\label{39}
T^{(eff)}_{ij}=T_{ij}+T^{(curve)}_{ij}
\end{align}
in which
\begin{align}\label{40}
T^{(curve)}_{ij}=\frac{1}{k^2}\big[\frac{(f(\mathcal{R})-\mathcal{R}f_{\mathcal{R}}(\mathcal{R}))}{2}g_{ij}
+(\nabla_{i}\nabla_{j}-g_{ij}{\square})f_{\mathcal{R}}(\mathcal{R})\big].
\end{align}
With the help of \eqref{38}, we get

\begin{align}\label{41}
\mathcal{R}_{ij}-\frac{\mathcal{R}}{2}g_{ij}=\frac{k^2}{f_\mathcal{R}(\mathcal{R})}T_{ij}^{(eff)}
\end{align}	
in which the equations \eqref{39} and  \eqref{40} are satisfied.\par

Now we establish the following:
\begin{lem}
 In $\mathcal{N}$, the followings are hold
 \begin{equation}\label{l1}
   \nabla_{k}\mathcal{R}=4\mathcal{R}A_{k},
 \end{equation}
 \begin{equation}\label{l2}
   \nabla_{l}\nabla_{k}\mathcal{R}=4\mathcal{R}f g_{lk}+ 4\mathcal{R}(4+f)A_{l}A_{k},
 \end{equation}
 \begin{equation}\label{l3}
   \square \mathcal{R}= 4\mathcal{R}(3f-4).
 \end{equation}
 \end{lem}
 \begin{proof}
 Let us choose the $\mathcal{N}$ spacetime. Then multiplying equation \eqref{1.6} by $g^{ij}$ and using equation \eqref{4.7}, we acquire
\begin{equation}\nonumber
  \nabla_{k}B=4A_{k}B
\end{equation}
 which entails
 \begin{equation}\label{l4}
   (a+4b)\nabla_{k}\mathcal{R}=\{4(a+4b)A_{k}-(a_{k}+4b_{k})\}\mathcal{R},
 \end{equation}
 where we have used equation \eqref{1.5}.\par
 
If we consider $a,b=$ constant, then equation \eqref{l4} yields equation \eqref{l1}.\par

Taking the covariant differentiation of equation \eqref{l1} and using the same equation, we can easily obtain \eqref{l2}.\par

Again, multiplying equation \eqref{l2} by $g^{lk}$, equation \eqref{l3} can be achieved.
 \end{proof}

Let $f(\mathcal{R})$ be an analytic function, then we write
\begin{equation}\label{11z}
\nabla_{i}\nabla_{j}f_{\mathcal{R}}(\mathcal{R})=f_{\mathcal{R}\mathcal{R}}(\mathcal{R})\nabla_{i}\nabla_{j}\mathcal{R}
+f_{\mathcal{R}\mathcal{R}\mathcal{R}}(\mathcal{R})(\nabla_{i}\mathcal{R})(\nabla_{j}\mathcal{R}).
\end{equation}
Hence, equation \eqref{11z} yields
\begin{equation}\label{12z}
  {\square}f_{\mathcal{R}}(\mathcal{R})=f_{\mathcal{R}\mathcal{R}}(\mathcal{R}){\square}\mathcal{R}
  +f_{\mathcal{R}\mathcal{R}\mathcal{R}}(\mathcal{R})g^{ij}(\nabla_{i}\mathcal{R})(\nabla_{j}\mathcal{R}).
\end{equation}
Now, using the equations \eqref{l1}, \eqref{l2} and \eqref{11z}, we acquire
\begin{eqnarray}\label{13z}
  \nabla_{i}\nabla_{j}f_{\mathcal{R}}(\mathcal{R}) &=& \{4\mathcal{R}(4+f)f_{\mathcal{R}\mathcal{R}}(\mathcal{R})
  +16\mathcal{R}^{2}f_{\mathcal{R}\mathcal{R}\mathcal{R}}(\mathcal{R})\}A_{i}A_{j}\nonumber\\&&
  +4 f \mathcal{R} f_{\mathcal{R}\mathcal{R}} (\mathcal{R}) g_{ij}.
\end{eqnarray}
Multiplying equation \eqref{13z} by $g^{ij}$, we infer
\begin{equation}\label{14z}
{\square}f_{\mathcal{R}}(\mathcal{R})=4\mathcal{R}(3f-4)f_{\mathcal{R}\mathcal{R}}(\mathcal{R})
-16\mathcal{R}^{2}f_{\mathcal{R}\mathcal{R}\mathcal{R}}(\mathcal{R}).
\end{equation}
Again, utilizing the equations \eqref{39}, \eqref{40} and \eqref{41}, we provide
\begin{align}\label{15z}
f_{\mathcal{R}}(\mathcal{R})\mathcal{R}_{ij}-f_{\mathcal{R}}(\mathcal{R})\frac{\mathcal{R}}{2}g_{ij}=k^2T_{ij}
-\frac{\mathcal{R}f_{\mathcal{R}}(\mathcal{R})-f(\mathcal{R})}{2}g_{ij}+(\nabla_{i}\nabla_{j}
-g_{ij}{\square})f_{\mathcal{R}}(\mathcal{R}).
\end{align}	
Using the equations \eqref{13z} and \eqref{14z} in \eqref{15z}, we reveal
\begin{eqnarray}\label{16z}
&&f_{\mathcal{R}}(\mathcal{R})\big(\mathcal{R}_{ij}-\frac{\mathcal{R}}{2}g_{ij}\big)=\kappa^2T_{ij}\nonumber\\&&
+\big[\frac{f(\mathcal{R})-\mathcal{R}f_{\mathcal{R}}(\mathcal{R})}{2}
+16\mathcal{R}^{2}f_{\mathcal{R}\mathcal{R}\mathcal{R}}(\mathcal{R})\nonumber\\&& +8\mathcal{R}(2-f)f_{\mathcal{R}\mathcal{R}}(\mathcal{R})\big]g_{ij}\nonumber\\&&
+\big[16\mathcal{R}^{2}f_{\mathcal{R}\mathcal{R}\mathcal{R}}(\mathcal{R}) +4\mathcal{R}(4+f)f_{\mathcal{R}\mathcal{R}}(\mathcal{R})\big]A_{i}A_{j}.
\end{eqnarray}	
Using equation \eqref{4.9} in equation \eqref{16z}, we acquire
\begin{eqnarray}\label{17z}
T_{ij}&=&-\frac{1}{\kappa^2}
\big[\frac{f(\mathcal{R})}{2}
+16\mathcal{R}^{2}f_{\mathcal{R}\mathcal{R}\mathcal{R}}(\mathcal{R})\nonumber\\&& +8\mathcal{R}(2-f)f_{\mathcal{R}\mathcal{R}}(\mathcal{R})+\frac{b}{a}\mathcal{R}f_{\mathcal{R}}(\mathcal{R}) \big]g_{ij}\nonumber\\&&
-\frac{1}{\kappa^2}\big[16\mathcal{R}^{2}f_{\mathcal{R}\mathcal{R}\mathcal{R}}(\mathcal{R}) +4\mathcal{R}(4+f)f_{\mathcal{R}\mathcal{R}}(\mathcal{R})\nonumber\\&&
 +\frac{(a+4b)}{a}\mathcal{R}f_{\mathcal{R}}(\mathcal{R})\big]A_{i}A_{j}.
\end{eqnarray}	
Therefore, using equations \eqref{1.2a} and \eqref{17z}, we find
\begin{eqnarray}\label{18z}
p&=&-\frac{1}{\kappa^2}
\big[\frac{f(\mathcal{R})}{2}
+16\mathcal{R}^{2}f_{\mathcal{R}\mathcal{R}\mathcal{R}}(\mathcal{R})\nonumber\\&& +8\mathcal{R}(2-f)f_{\mathcal{R}\mathcal{R}}(\mathcal{R})+\frac{b}{a}\mathcal{R}f_{\mathcal{R}}(\mathcal{R}) \big]
\end{eqnarray}
and
\begin{eqnarray}\label{19z}
\sigma&=&-\frac{1}{\kappa^2}
\big[-\frac{f(\mathcal{R})}{2}-\frac{b}{a}\mathcal{R}f_{\mathcal{R}}(\mathcal{R})\nonumber\\&& -12\mathcal{R}f f_{\mathcal{R}\mathcal{R}}(\mathcal{R})
+\frac{(a+4b)}{a}\mathcal{R}f_{\mathcal{R}}(\mathcal{R}) \big]
\end{eqnarray}
Therefore, we provide:
\begin{theo}
In a spacetime $\mathcal{N}$ satisfying $f(\mathcal{R})$ gravity, $p$ and $\sigma$ are given by equations \eqref{18z} and \eqref{19z}, respectively.
\end{theo}

\subsection{Energy Conditions}
In the following subsection, we verify the energy conditions for the new model $f\left(\mathcal{R}\right)= e^{(\alpha \mathcal{R})}-ln(\beta \mathcal{R})$ in which $\alpha$ and $\beta$ are positive constants.\par

In modified gravity the energy conditions are demonstrated as
\begin{align*}
 \mathrm{NEC}&\quad\mathrm{if\;and\;only\;if}\quad\sigma+p\geq0,\\
  \mathrm{SEC}&\quad\mathrm{if\;and\;only\;if}\quad\sigma+p\geq0\quad\mathrm{and}\quad\sigma+3p\geq0,\\
 \mathrm{DEC}&\quad\mathrm{if\;and\;only\;if}\quad\sigma\pm p\geq0\quad\mathrm{and}\quad\sigma\geq0,\\
 \mathrm{WEC}&\quad\mathrm{if\;and\;only\;if}\quad\sigma+p\geq0\quad\mathrm{and}\quad\sigma\geq0.
\end{align*}
With the help of \eqref{18z} and \eqref{19z}, $p$ and $\sigma$ are given by
\begin{eqnarray}\label{5.11}
p&=&-\frac{1}{\kappa^2}
\big[\frac{e^{(\alpha \mathcal{R})}-ln(\beta \mathcal{R})}{2}+16\mathcal{R}^{2}(\alpha^{3}e^{(\alpha \mathcal{R})}-\frac{2}{\mathcal{R}^{3}})\nonumber\\&&
+8\mathcal{R}(2-f)(\alpha^{2}e^{(\alpha \mathcal{R})}+\frac{1}{\mathcal{R}^{2}})+\frac{b}{a}\mathcal{R}(\alpha e^{(\alpha \mathcal{R})}-\frac{1}{\mathcal{R}}) \big]
\end{eqnarray}
and
\begin{eqnarray}\label{5.12}
\sigma&=&-\frac{1}{\kappa^2}
\big[-\frac{e^{(\alpha \mathcal{R})}-ln(\beta \mathcal{R})}{2}+\frac{(a+4b)}{a}\mathcal{R}(\alpha e^{(\alpha \mathcal{R})}-\frac{1}{\mathcal{R}})\nonumber\\&&
+12\mathcal{R}f(\alpha^{2}e^{(\alpha \mathcal{R})}+\frac{1}{\mathcal{R}^{2}})-\frac{b}{a}\mathcal{R}(\alpha e^{(\alpha \mathcal{R})}-\frac{1}{\mathcal{R}}) \big].
\end{eqnarray}

The energy conditions for the above model are now examined. The energy conditions for this arrangement may now be discussed using equations \eqref{5.11} and \eqref{5.12}. To draw the following figures, we set $\kappa=2.077\times 10^{-43}$, $f=a=b=-1$, $\alpha=1$, $\beta \in [1,2]$ and $R \in [0,1]$.\par
\begin{tabulary}{\linewidth}{CC}
	\includegraphics[height=0.24\textheight]{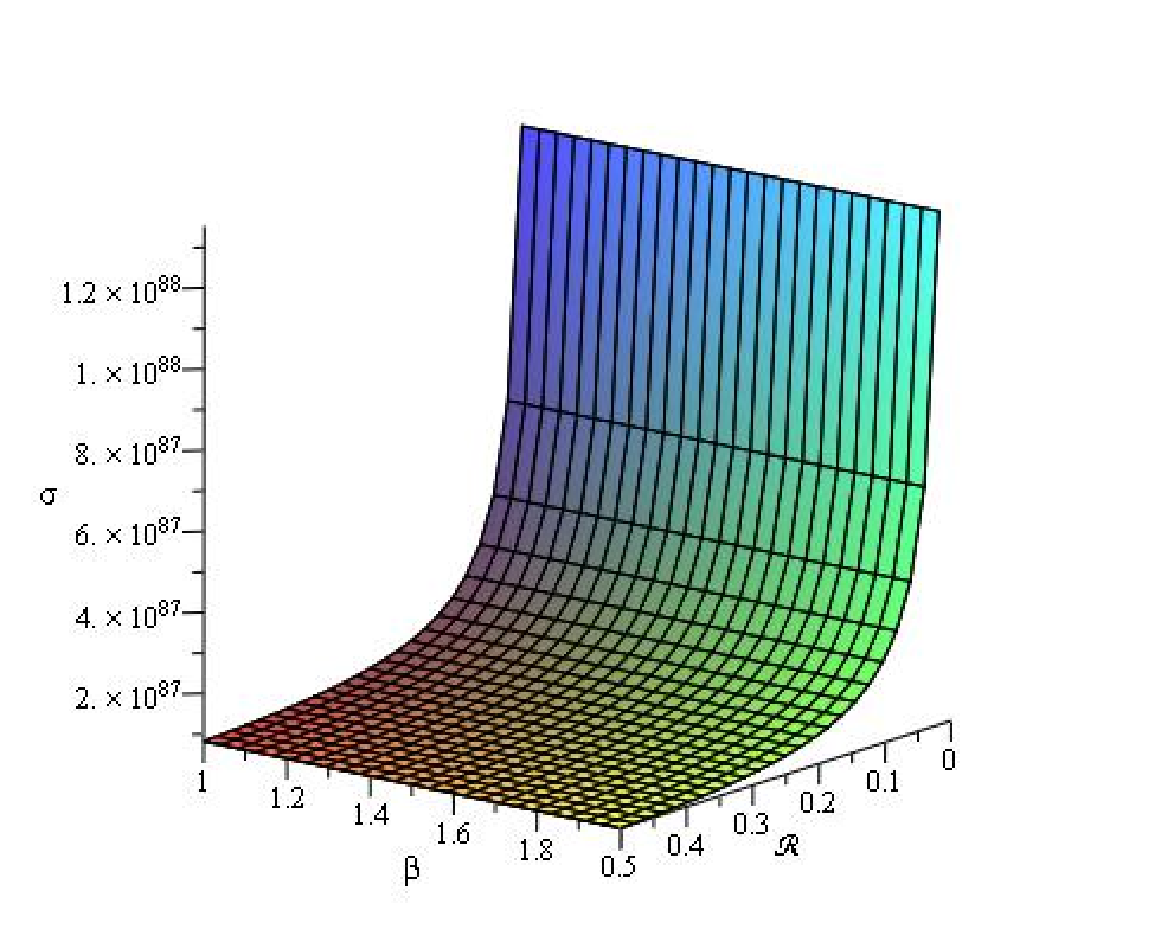}
	&
	\includegraphics[height=0.24\textheight]{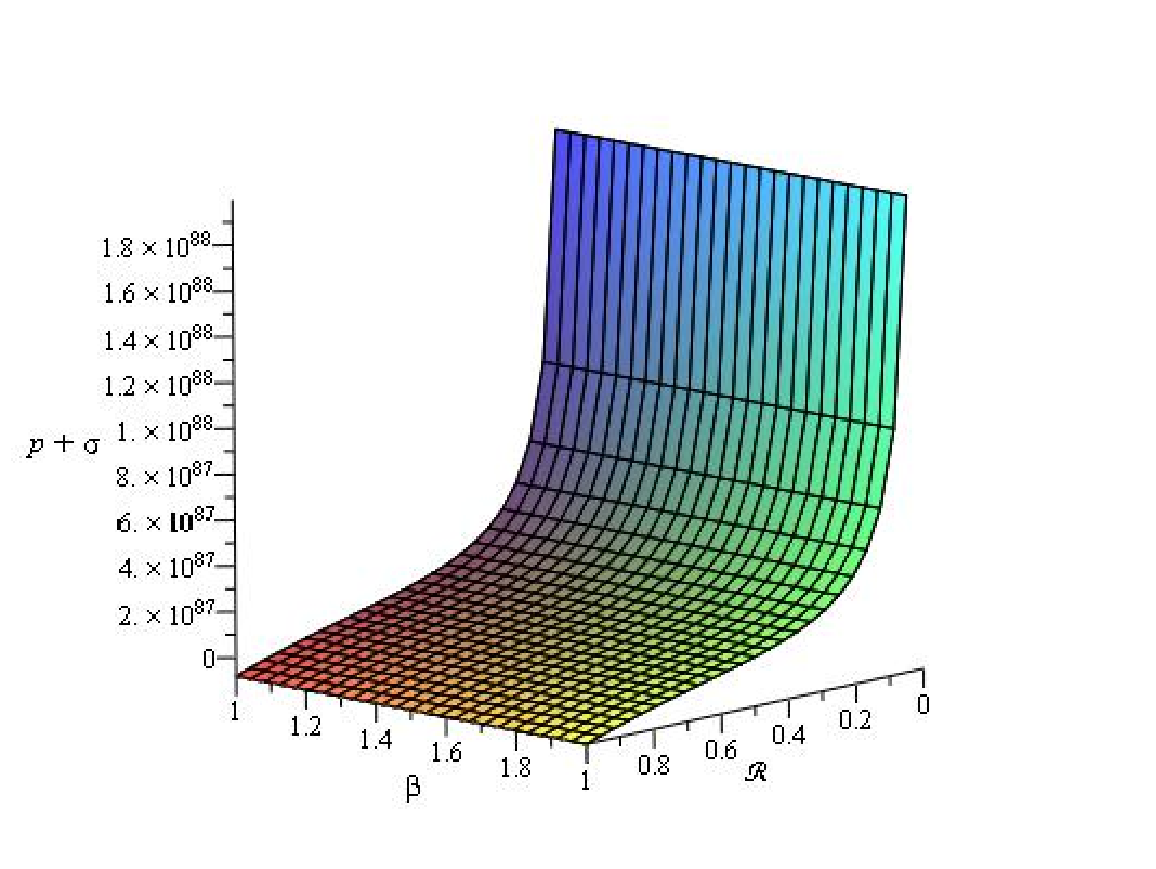}
	\\
	{\bf Fig. 1:} Development of $\sigma$ with reference to $\mathcal{R}$ and $\beta$ &{\bf Fig. 2:} Development of $p+\sigma$ with reference to $\mathcal{R}$ and $\beta$
	
\end{tabulary}
\begin{tabulary}{\linewidth}{CC}
	
	\includegraphics[height=0.25\textheight]{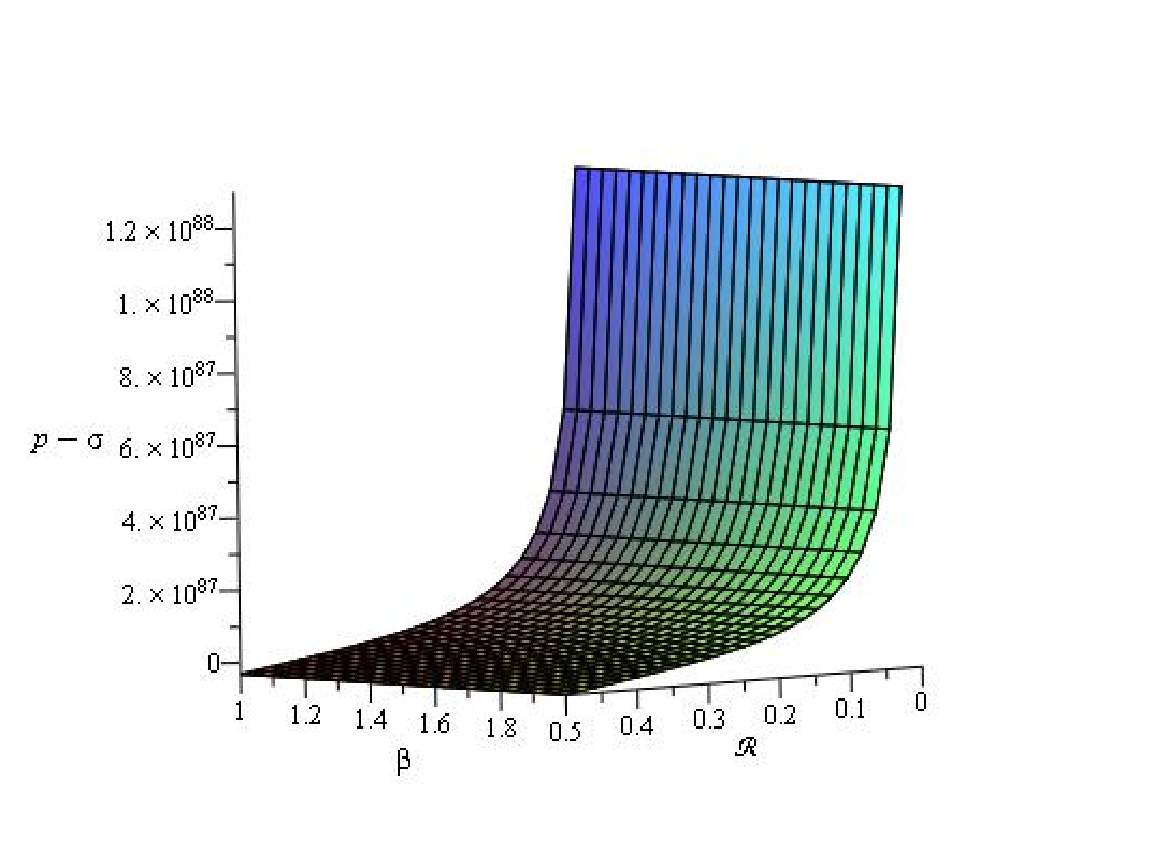}
    &
	\includegraphics[height=0.24\textheight]{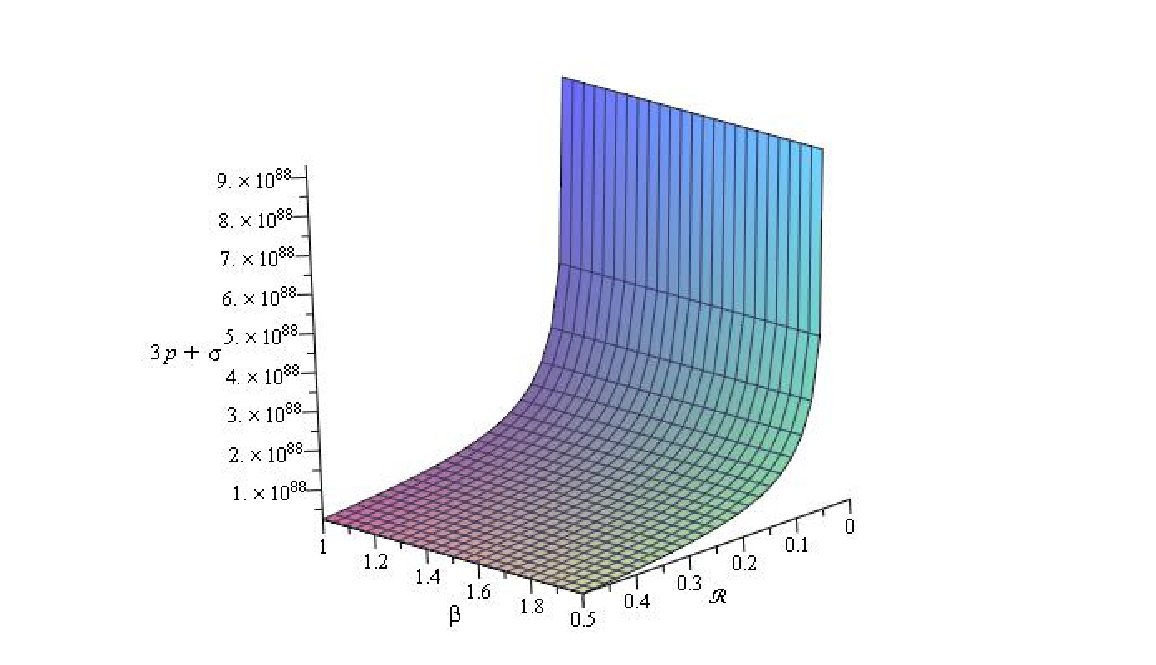}
    \\	
	
	{\bf Fig. 3:} Development of $p-\sigma$ with reference to $\mathcal{R}$ and $\beta$  &{\bf Fig. 4:} Development of $\sigma+3p$ with reference to $\mathcal{R}$ and $\beta$
	
\end{tabulary}

Figures $1$ and $2$ demonstrate that, for parameters $\beta \in\left[1,2\right], \mathcal{R}\in\left[0,1\right]$, the energy density and $p+\sigma$ can not be negative, and that, for larger values of $\beta$ and $\mathcal{R}$, they are high. Because NEC is a component of WEC, NEC and WEC are fulfilled. The $p-\sigma$ profile for $\beta \in\left[1,2\right], \mathcal{R}\in\left[0,1\right]$ is positive, as seen in Fig. $3$. It is evident from Figs. $1$, $2$, and $3$ that DEC is validated. Furthermore, we can observe that SEC is satisfied from Figs. $2$ and $4$, and this finding yields the late-time acceleration of the Cosmos\cite{LDMS21}. Moreover, each result aligns with the $\Lambda$CDM model \cite{AAAABBBBBBB20}.

\subsection{$\mathcal{N}$ spacetime in $f(\mathcal{R})$ gravity with constant $\mathcal{R}$:}

The $f(\mathcal{R})$ gravity field equations are as follows:
\begin{eqnarray}
  \kappa T_{lk} &=& f^{\prime}(\mathcal{R})\mathcal{R}_{lk}
  -f^{\prime\prime\prime}(\mathcal{R})\nabla_{l}\mathcal{R}\nabla_{k}\mathcal{R}
  -f^{\prime\prime}(\mathcal{R})\nabla_{l}\nabla_{k}\mathcal{R}\nonumber\\&&
  +g_{lk}[f^{\prime\prime\prime}(\mathcal{R})\nabla_{m}\mathcal{R}\nabla^{m}\mathcal{R}
  +f^{\prime\prime}(\mathcal{R})\nabla^{2}\mathcal{R}-\frac{1}{2}f(\mathcal{R})],
\end{eqnarray}
$f^{\prime}(\mathcal{R})$ is the derivative with respect to $\mathcal{R}$.\par
For $\mathcal{R}=$ constant, we infer
\begin{equation}\label{d2}
 \mathcal{R}_{lk}-\frac{f}{2f^{\prime}}g_{lk}=\frac{\kappa}{f^{\prime}}T_{lk}.
\end{equation}
Using equation \eqref{1.2a} in equation \eqref{d2}, we obtain
\begin{equation}\label{d3}
 \mathcal{R}_{lk}=\frac{f}{2f^{\prime}}g_{lk}+\frac{\kappa}{f^{\prime}}[(\sigma+p)u_{k}u_{l}+p g_{kl}].
\end{equation}
Making use of equations \eqref{4.9} and \eqref{d3}, we acquire
\begin{equation}\label{d4}
 \frac{\kappa}{f^{\prime}}(\sigma+p)=\frac{a+4b}{a}\mathcal{R}
 \end{equation}
and
\begin{equation}\label{d5}
\frac{p\kappa}{f^{\prime}}+ \frac{f}{2f^{\prime}}= -\frac{b}{a}\mathcal{R}.
\end{equation}
Solving the last two equations, we get
\begin{equation}\label{d6}
  p=-\frac{bf^{\prime}\mathcal{R}}{\kappa a}-\frac{f}{2\kappa}
\end{equation}
and
\begin{equation}\label{d7}
  \sigma=-\frac{(a+3b)f^{\prime}\mathcal{R}}{\kappa a}+\frac{f}{2\kappa}.
\end{equation}
Therefore, we provide:
\begin{theo}
In a spacetime $\mathcal{N}$ obeying $f(\mathcal{R})$ gravity with constant $\mathcal{R}$, $p$ and $\sigma$ are given by \eqref{d6} and \eqref{d7}, respectively.
\end{theo}

\subsection{Energy Conditions}
In general relativity, energy conditions are vital tools to study black holes and wormholes in numerous modified gravities. To specify certain energy conditions in our current investigation of the $f(\mathcal{R})$ gravity, we must find the effective isotropic pressure $p^{eff}$ and the effective energy density $\sigma^{eff}$ (see, \cite{ade}-\cite{ade1}).\par

Equation \eqref{d2}, can be rewritten in the following form
\begin{equation}\label{d8}
 \mathcal{R}_{lk}-\frac{\mathcal{R}}{2}g_{lk}=\frac{\kappa}{f^{\prime}}T_{lk}^{eff},
\end{equation}
in which
\begin{equation}\label{d9}
  T_{lk}^{eff}=T_{lk}+\frac{(f-\mathcal{R}f^{\prime})}{2\kappa}g_{lk}.
\end{equation}
Then equation \eqref{1.2a} reduces to
\begin{equation}
\label{d10}
T_{lk}^{eff}=(\sigma^{eff}+p^{eff})a_{l}a_{k}+p^{eff} g_{lk},
\end{equation}
in which $p^{eff}=p+\frac{(f-\mathcal{R}f^{\prime})}{2\kappa}$ and $\sigma^{eff}=\sigma-\frac{(f-\mathcal{R}f^{\prime})}{2\kappa}$.\par
In our case, using equations \eqref{d6} and \eqref{d7}, we provide
\begin{equation}\label{d11}
  p^{eff}=-\frac{bf^{\prime}\mathcal{R}}{\kappa a}-\frac{\mathcal{R} f^{\prime}}{2\kappa}
\end{equation}
and
\begin{equation}\label{d12}
  \sigma^{eff}=-\frac{(a+3b)f^{\prime}\mathcal{R}}{\kappa a}+\frac{\mathcal{R} f^{\prime}}{2\kappa}.
\end{equation}
Now we investigate the energy conditions in $\mathcal{N}$ with non-zero constant scalar curvature
obeying $f(\mathcal{R})$ gravity. With the help of equations \eqref{d11} and \eqref{d12}, we find the conditions for NEC, DEC, WEC and SEC, respectively in this set up and they are given as follows:

\begin{table}[h]
    \centering
        \begin{tabular}{|c|c|c|}
    \hline
    \multicolumn{3}{|c|}{Table 1}\\
    \hline
     \multicolumn{3}{|c|}{Validity of energy conditions in $(PBS)_{4}$ spacetime }\\
    \hline
      Energy Condition & Inequalities & Conditions of validation\\
        \hline
        NEC & $p^{eff}+\sigma^{eff} \geq 0$ & $ (a+4b)\mathcal{R} \leq 0$\\
        \hline
        WEC & $\sigma^{eff} \geq 0 $ and  $p^{eff}+\sigma^{eff} \geq 0$ & $(a+6b)\mathcal{R} \leq 0$ \\
         & & and $ (a+4b)\mathcal{R} \leq 0$\\
        \hline
        DEC & $\sigma^{eff} \geq 0 $ and $\sigma^{eff} \pm p^{eff} \geq 0$ & $(a+6b)\mathcal{R} \leq 0$ , $ b\mathcal{R} \geq 0$ \\
         & & and $(a+4b)\mathcal{R} \leq 0$\\
        \hline
        SEC& $\sigma^{eff} \geq 0 $ and $\sigma^{eff} +3p^{eff} \geq 0$ & $(a+6b)\mathcal{R} \leq 0$ \\
         & & and $ (a+2b)\mathcal{R} \leq 0$ \\
              \hline
              \end{tabular}
\end{table}

\section{Discussion}
The physical inspiration for researching numerous spacetime models in cosmology is to gain additional insight into particular phases of the universe's evolution, which can be split in the following ways:\\
(i) The initial, (ii) The intermediate and (iii) The final stage.\par
The first stage is concerned with viscous fluid and the second with non-viscous fluid While both stages admit heat flux. The last stage, which has thermal equilibrium, describes the ideal fluid stage. We consider the last stage in our study, and it is illustrated that $\mathcal{N}$ represents a PFS.\par

We investigate the $f(R)$ gravity model with the geometric restriction of $(PBS)_{4}$ spacetime, and find that $p$ and $\sigma$ are not constants. Therefore, we could say that the current universe is consistent with the $(PBS)_{4}$ spacetime.\par

The initial $f(\mathcal{R})$-model, $f(\mathcal{R}) = \mathcal{R} +\alpha \mathcal{R}^{2}$, ($\alpha>0$) presented by Starobinsky \cite{aas} was aimed at explaining cosmic inflation as a pure gravitational effect without the use of dark energy. Carroll et al. \cite{car} presented the model $f(\mathcal{R}) = \mathcal{R}-\frac{\mu^{4}}{\mathcal{R}}$, ($\mu > 0$) to explain late-time acceleration as a scalar field. Despite their limitations, these models were still able to popularize $f(\mathcal{R})$-models in general. In \cite{LDMS21}, by choosing the model $f(\mathcal{R}) = \mathcal{R}-\alpha (1-e ^{-\frac{\mathcal{R}}{\alpha}})$ the authors have shown that NEC, WEC and DEC have been satisfied, whereas SEC violated. Here, our findings have been assessed both analytically and graphically. Our formulation was constructed using the analytical technique, and one cosmological model, $f\left(\mathcal{R}\right)=e^{(\alpha \mathcal{R})}-ln(\beta \mathcal{R})$, was evaluated for stability. For this model, we find that NEC, WEC, DEC and SEC are satisfied. Also, this finding yields the late-time acceleration of the Cosmos\cite{LDMS21} and each result aligns with the $\Lambda$CDM model \cite{AAAABBBBBBB20}.\par

In \cite{cap1}, it is deduced that a GRW spacetime of dimension $n$ with $\nabla_{k}C^{k}_{lij} =0$ reveals a perfect fluid type energy momentum tensor for any $f(\mathcal{R})$ gravity model. Hence, from Theorem \ref{th3.4}, we state
the following:\par
A $(PBS)_{4}$ $\mathrm{GRW}$ spacetime represents a perfect fluid type energy momentum tensor for any model of $f(\mathcal{R})$ gravity.\par

 A collective study of validation of energy conditions are examined and the outcome is mentioned in Table 1. It is seen that the presence of exotic matter is not required for our case.

\section{Declarations}
\subsection{Funding }
NA.
\subsection{Code availability}
NA.
\subsection{Availability of data}
NA.
\subsection{Conflicts of interest}
The authors affirm that they do not have any competing interests.

\section{Acknowledgement}
The authors are thankful to the referee for his or her valuable suggestions towards the improvement of the paper.

\end{document}